\documentclass[11pt]{article}
\usepackage{fullpage}
\usepackage[utf8]{inputenc}
\usepackage{algorithm}
\usepackage{algorithmic}
\usepackage{amsfonts}
\usepackage{graphicx}
\usepackage{amsthm}
\usepackage{mathrsfs}
\usepackage{amssymb}
\usepackage{hyperref}
\usepackage{amsmath}
\usepackage{cite}
\usepackage{xcolor}
\usepackage{comment}
\usepackage{mathtools}

\newtheorem{requirement}{}
\theoremstyle{plain}
\newtheorem{theorem}{Theorem}[section]
\newtheorem{lemma}[theorem]{Lemma}
\newtheorem{corollary}[theorem]{Corollary}
\newtheorem{observation}[theorem]{Observation}
\theoremstyle{definition}

\newcommand{\gio}{G_\mathsf{i,o}}
\newcommand{\prdiv}{\mathtt{prd}(\msfi(v))}
\newcommand{\dataiv}{\mathtt{data}(\msfi(v))}
\newcommand{\true}{\mathtt{True}}
\newcommand{\false}{\mathtt{False}}
\newcommand{\gi}{G_\mathsf{i}}
\newcommand{\pif}{\langle \Pi,f\rangle}
\newcommand{\msfi}{\mathsf{i}}
\newcommand{\msfo}{\mathsf{o}}
\newcommand{\Vfunc}{V\rightarrow\{0,1\}^*}

\newcommand{\calu}{\mathcal{U}}
\newcommand{\calf}{\mathcal{F}}
\newcommand{\legpi}{\mathcal{LEG}(\Pi)}

\newcommand{\inner}{\mathtt{inner}}
\newcommand{\rim}{\mathtt{rim}}
\newcommand{\partopt}{\mathtt{Part\_OPT}}
\newcommand{\partcgf}{\mathtt{Part\_CGF}}
\begin{document}
	\title{Locally Restricted Proof Labeling Schemes (Full Version)\thanks {This work was supported in part by the Technion Hiroshi Fujiwara Cyber Security Research Center and the Israel National Cyber Directorate. In addition, the work of Shay Kutten was also supported in part by ISF grant 1346/22.}}
	\author{Yuval Emek\footnote{Technion - Israel Institute of Technology. yemek@technion.ac.il} \and Yuval Gil\footnote{Technion - Israel Institute of Technology. yuval.gil@campus.technion.ac.il} \and Shay Kutten \footnote {Technion - Israel Institute of Technology. kutten@technion.ac.il}
	}
	\date{}
	
	\maketitle		

	\begin{abstract}
		Introduced by Korman, Kutten, and Peleg (PODC 2005), a \emph{proof labeling
			scheme (PLS)} is a distributed verification system dedicated to evaluating if
		a given \emph{configured graph} satisfies a certain property.
		It involves a centralized \emph{prover}, whose role is to provide proof that a
		given configured graph is a yes-instance by means of assigning \emph{labels}
		to the nodes, and a distributed verifier, whose role is to verify the validity
		of the given proof via local access to the assigned labels.
		In this paper, we introduce the notion of a \emph{locally restricted} PLS in
		which the prover's power is restricted to that of a LOCAL algorithm with a
		polylogarithmic number of rounds.
		To circumvent inherent impossibilities of PLSs in the locally restricted
		setting, we turn to models that relax the correctness requirements by allowing
		the verifier to accept some no-instances as long as they are not ``too far''
		from satisfying the property in question.
		To this end, we evaluate
		(1)
		\emph{distributed graph optimization problems (OptDGPs)} based on the notion
		of an \emph{approximate proof labeling scheme (APLS)} (analogous to the type of
		relaxation used in sequential approximation algorithms);
		and
		(2) \emph{configured graph families (CGFs)} based on the notion of a
		\emph{testing proof labeling schemes (TPLS)} (analogous to the type of
		relaxation used in property testing algorithms).
		The main contribution of the paper comes in the form of two generic compilers,
		one for OptDGPs and one for CGFs:
		given a black-box access to an APLS (resp., PLS) for a large class of OptDGPs
		(resp., CGFs), the compiler produces a locally restricted APLS (resp., TPLS)
		for the same problem, while losing at most a
		$(1 + \epsilon)$
		factor in the scheme's relaxation guarantee.
		An appealing feature of the two compilers is that they only require a
		logarithmic additive label size overhead. 
	\end{abstract}
	
	\section{Introduction}
	\label{section:introduction}
	A \emph{proof system} is a tool designed to verify the correctness of a
	certain claim.
	It is composed of two entities:
	a \emph{prover}, whose role is to provide proof for the claim in question;
	and a computationally bounded \emph{verifier} that seeks to verify the validity
	of the given proof.
	The crux of a proof system is that the proof given by the prover cannot be
	blindly trusted.
	That is, for a proof system to be correct, the verifier must be able to
	distinguish between an honest prover, providing a correct proof, and a
	malicious prover who tries to convince the verifier of a false claim.
	
	In the realm of \emph{distributed computing}, the study of proof systems, also
	known as \emph{distributed proof systems}, has gained a lot of attention.
	The goal of a distributed proof system is to decide if a given
	\emph{configured graph} satisfies a certain property.
	This is typically done by means of a centralized prover that has a global view
	of the entire configured graph, and a distributed verifier that operates at
	all nodes concurrently and is subject to locality restrictions.
	Various models for distributed proof systems have been introduced in the
	literature, including \emph{proof labeling schemes (PLSs)} \cite{pls},
	\emph{locally checkable proofs} \cite{lcp}, \emph{nondeterministic local
		decisions} \cite{nld1}, and \emph{distributed interactive proofs}
	\cite{Kol-interactive-proofs}.
	
	The current paper focuses on the PLS model, introduced by Korman, Kutten, and
	Peleg \cite{pls} (see Section~\ref{section:pls} for the formal definition).
	In a PLS, the prover generates its proof by means of assigning a \emph{label}
	to each node.
	The verification process performed by the verifier at each node $v$ has access
	to $v$'s label and to the labels of $v$'s neighbors, but it cannot access the
	labels assigned to nodes outside its local neighborhood. 
	The correctness requirements state that if the given configured graph is a
	yes-instance, then all nodes must accept;
	and if the given configured graph is a no-instance, then at least one node
	must reject.
	The standard performance measure of a PLS is its \emph{proof size}, defined to
	be the size of the largest label assigned by the honest prover.  
	
	Recently, there is a growing interest from (sequential) computational
	complexity researchers in \emph{doubly efficient} proof systems
	\cite{GoldwasserRK17, ReingoldRR21}.
	These proof systems are characterized by restricting the (honest) prover to
	``efficient computations'' --- i.e., polynomial time algorithms --- on top of
	the restrictions imposed on the computational power of the (still weaker)
	verifier.
	For example, Goldwasser et al.~\cite{GoldwasserRK17} consider polynomial
	time provers vs.\ logarithmic space verifiers, whereas Reingold et
	al.~\cite{ReingoldRR21} consider polynomial time provers vs.\ linear time
	and near-linear time verifiers.
	
	Motivated by the success story of doubly efficient proof systems in sequential
	computational complexity, in this paper, we initiate the study of such proof
	systems in the distributed computing realm.
	To do so, we adjust the notion of ``efficient computations'' from sequential
	algorithms running in polynomial time to LOCAL algorithms running in a
	polylogarithmic number of rounds \cite{Peleg2000}.
	This introduces a new type of PLSs, called \emph{locally restricted}
	PLSs, where the label assigned to a node $v$ is computed by the (honest)
	prover based on the subgraph induced by the nodes within polylogarithmic 
	distance from $v$, rather than the whole graph (refer to
	Section~\ref{section:locally-restricted} for a formal definition).
	
	Beyond the theoretical interest that lies in this new type of distributed
	proof systems, we advocate for their investigation also from a more practical
	point of view:
	A natural application of PLSs is local checking for self-stabilizing
	algorithms \cite{Awerbuch-1991} which involves a detection module and a
	correction module.
	In this mechanism, the verifier's role is played by the detection module and
	the prover's role is played by a dedicated sub-module of the correction module
	responsible for the label assignment to the nodes \cite{pls} (the correction
	module typically includes another sub-module, responsible for constructing the
	actual solution, which is abstracted away by the PLS).
	Since both modules operate as distributed algorithms, any attempt to implement
	them in practice should take efficiency considerations into account.
	While classic PLSs consider this efficiency requirement (only) from the
	verifier's point of view, in locally restricted PLSs, we impose efficiency
	demands on both the verifier and the prover.

	
	It turns out that locally restricted PLSs are impossible for many interesting
	properties, regardless of proof size (as shown in the simple observation presented in
	Section~\ref{section:impossibilities}).
	This leads us to slightly relax the correctness requirements of a PLS so that
	the verifier may also accept no-instances as long as they are not ``too far''
	from satisfying the property in question.
	Specifically, we consider locally restricted schemes in the context of two
	relaxed models called \emph{approximate proof labeling schemes (APLS)}
	\cite{aplspaper, EmekG20}
	and \emph{testing proof labeling schemes (TPLS)}. 
	
	The APLS model was introduced by Censor-Hillel et al.~\cite{aplspaper} and
	studied further by Emek and Gil \cite{EmekG20}.
	For an approximation parameter
	$\alpha \geq 1$,
	the goal of an $\alpha$-APLS for a \emph{distributed graph optimization
		problem (OptDGP)} is to distinguish between optimal instances and instances
	that are $\alpha$-far from being optimal (refer to Section~\ref{section:model}
	for the definitions). Interestingly, for some classical edge-based \emph{covering/packing} OptDGPs (e.g., maximum matching and minimum edge cover), locally restricted APLSs are already established in previous works \cite{lcp,aplspaper,EmekG20}. In contrast, until the current paper, known APLSs for node-based covering/packing OptDGPs require the prover to have a global view of the given configured graph (see, e.g., the APLS for minimum weight vertex cover presented in \cite{EmekG20}). 
	In Section \ref{section:compiler-optdgps}, we develop a generic compiler that
	gets a (not necessarily locally restricted) $\alpha$-APLS for an OptDGP $\Psi$,
	belonging to a large class of  node-based covering/packing OptDGPs, and generates a
	locally restricted
	$((1 + \epsilon) \alpha)$-APLS
	for $\Psi$, where $\epsilon$ is a constant performance parameter.
	The proof size of the locally restricted
	$((1+\epsilon) \alpha)$-APLS
	generated by our compiler is
	$\ell_{\Psi, \alpha} + O (\log n)$,
	where
	$\ell_{\Psi, \alpha}$
	is the proof size of the $\alpha$-APLS provided to the compiler. Refer to Section \ref{section:high-level} for a high-level overview of this construction.
	
	The TPLS model is developed in the current paper based on the notion of
	\emph{property testing}
	\cite{GoldreichGR98, AlonKKR08}.
	For a parameter
	$\delta > 0$,
	the goal of a $\delta$-TPLS for a \emph{configured graph family (CGF)} $\Phi$
	is to distinguish between configured graphs belonging to $\Phi$ and configured
	graphs that are $\delta$-far from belonging to $\Phi$, where the distance here
	is measured in terms of the graph topology.
	In Section \ref{section:compiler-cgfs}, we develop a generic compiler that
	gets a (not necessarily locally restricted) PLS for a CGF $\Phi$, that is
	closed under node-induced subgraphs and disjoint union, and generates a
	locally restricted $\delta$-TPLS for $\Phi$.
	The proof size of the locally restricted $\delta$-TPLS generated by our
	compiler is
	$\ell_{\Phi} + O (\log n)$,
	where $\ell_{\Phi}$ is the proof size of the PLS provided to the compiler. Refer to Section \ref{section:cgf-high-level} for a high-level overview of this construction.
	
	The applicability of our compilers is demonstrated in Section
	\ref{section:bounds}, where we show
	how the two compilers can be used to obtain APLSs and TPLSs for various
	well-known OptDGPs and CGFs, respectively;
	refer to Tables \ref{apls-table} and \ref{tpls-table} for a summary of these
	results.
	We conclude with additional related work presented in
	Section~\ref{section:related}.

	\begin{table}
		\begin{center}
			\begin{tabular}{| l| l| l| l| }
				\hline
				\textbf{OptDGP} & \textbf{Graph family} & \textbf{Approx.\  ratio} & \textbf{Proof size} \\
				\hline
				
				minimum weight vertex cover & any & $2(1+\epsilon)$&$O
				(\log n)$ \\
				\hline
				minimum vertex cover & odd-girth $=\omega(\log n)$ & $1+\epsilon$&$O
				(\log n)$\\
				\hline
				maximum independent set & any & $\Delta(1+\epsilon)$&$O
				(\log n)$\\
				
				& odd-girth $=\omega(\log n)$ & $1+\epsilon$&$O
				(\log n)$\\
				\hline
				minimum weight dominating set& any & $O(\log n)$&$O
				(\log n)$\\
				\hline
				any canonical OptDGP & any & $1+\epsilon$&$O
				(n^2)$\\
				\hline
			\end{tabular}
		\end{center}
		\caption{Locally restricted $\alpha$-APLS results and proof sizes.}
		\label{apls-table}
	\end{table}
	
	\begin{table}
		\begin{center}
			\begin{tabular}{| l| l| }
				\hline
				\textbf{CGF} & \textbf{Proof size} \\
				\hline
				planarity & $O(\log n)$\\
				\hline
				bounded arboricity  & $O(\log n)$\\
				\hline
				$k$-colorability  & $O(\log n)$\\
				\hline
				forest  & $O(\log n)$\\
				\hline
				DAG  & $O(\log n)$\\
				\hline
			\end{tabular}
		\end{center}
		\caption{Locally restricted $\delta$-TPLS results and proof sizes.}
		\label{tpls-table}
	\end{table}
	
	\subsection{Paper's Organization}
	In Section \ref{section:model}, we present the model. Preliminaries are presented in Section \ref{section:preliminaries}. Following that, in Sections \ref{section:compiler-optdgps} and \ref{section:compiler-cgfs}, we present our compiler for OptDGPs and CGFs, respectively. Within these sections, a high-level overview of the compilers appears in Subsections \ref{section:high-level} and \ref{section:cgf-high-level}. In Section \ref{section:bounds}, we present concrete bounds for APLS and TPLS construction. In Section \ref{section:impossibilities}, we discuss the impossibilities of locally restricted PLSs for a variety of problems. Finally, additional related work is presented in Section \ref{section:related}. 
	
	\section{Model}
	\label{section:model}
	We consider \emph{distributed verification systems} in which evaluated instances are called \emph{configured graphs}. A configured graph $G_s=\langle G,s\rangle $ is a pair consisting of an undirected graph $G=(V,E)$ and a \emph{configuration function} $s:\Vfunc$ that assigns a bit string $s(v)$, referred to as $v$'s \emph{local configuration}, to each node $v\in V$. Throughout this paper, we use the notation $n=|V|$ and $m=|E|$.

	For a node $v\in V$, we stick to the convention that $N_G(v) = \{ u \mid (u, v) \in E  \}$ 
	denotes the set of $v$'s \emph{neighbors} in $G$ and that $\deg_{G}(v)=|N_G(v)|$ denotes $v$'s \emph{degree} in $G$. When $G$ is clear from the context, we may omit it from our notations and use $N(v)$ and $\deg(v)$ instead of $N_G(v)$ and $\deg_G(v)$, respectively.  
	
	We assume that all configured graphs considered in the context of this paper are \emph{identified}, i.e., the configuration function $s:\Vfunc$ assigns a unique id of size $O(\log n)$, denoted by $id(v)$, to each node $v\in V$. Moreover, we assume that the local configuration $s(v)$ distinguishes between node $v$'s incident edges by means of a set $\mathcal{A}(v)$ of abstract \emph{port names}, and a bijection $\rho^{s}_v:N(v)\rightarrow\mathcal{A}(v)$, referred to as the \emph{internal port name assignment} of $v$, that assigns a (locally unique) port name $\rho_{v}^{s}(u)$ to each node $u \in N(v)$. More concretely, assume that the local configuration $s(v)$ includes a designated field for each neighbor $u \in N(v)$ and that this field is indexed by $\rho_{v}^{s}(u)$. Unless stated otherwise, when we refer to an ordered list $u_1,\dots,u_{\deg(v)}$ of $v$'s neighbors, it is assumed that the list is ordered by $v$'s internal port name assignment.
	
	Given a configured graph $G_s$ consisting of graph $G=(V,E)$ and configuration function $s$, we say that a configured graph $G'_{s'}$ consisting of graph $G'=(V',E')$ and configuration function $s'$, is a \emph{configured subgraph} of $G_s$ if (1) $G'$ is a subgraph of $G$, i.e., $V'\subseteq V$ and $E'\subseteq E$; and (2) the configuration function $s'$ is the projection of $s$ on $G'$, where for each node $v \in V'$, the fields corresponding to nodes $u \in N_{G}(v) \setminus N_{G'}(v)$ are omitted from the local configuration $s'(v)$ and the internal port name assignment $\rho_v^{s'}$ associated with $s'$ is defined so that $\rho_v^{s'}(u)=\rho_v^{s}(u)$ for each $u\in N_{G'}(v)$. For a node subset $U\subseteq V$, let $G(U)$ denote the subgraph induced on $G$ by $U$ and let $G_s(U)$ be the configured subgraph of $G_s$ defined over the subgraph $G(U)$.
	
	We define a \emph{configured graph family (CGF)} as a collection of configured graphs.\footnote{Refer to Table \ref{abbreviations} for a full list of the abbreviations used in this paper} A CGF type that plays a central role in this paper is that of a \emph{distributed graph problem (DGP)} $\Pi$, where for each configured graph $G_s\in \Pi$, the configuration function $s$ is composed of an \emph{input assignment} $\msfi:\Vfunc$ and an \emph{output assignment} $\msfo:\Vfunc$. We refer to such a configured graph as an \emph{input-output (IO) graph} and often denote it by $\gio$. The input assignment $\msfi$ assigns to each node $v\in V$, a bit string $\msfi(v)$, referred to as $v$'s \emph{local input}, that encodes attributes associated with $v$ and its incident edges (e.g., node ids, edge orientations, edge weights, and node weights); whereas the output assignment $\msfo$ assigns a \emph{local output} $\msfo(v)$ to each node $v\in V$. For an input assignment $\msfi$, we refer to the configured graph $\gi=\langle G,\msfi\rangle$ as an \emph{input graph}.
	
	Consider a DGP $\Pi$. An input graph $\gi$ is said to be \emph{legal} (and the graph $G$ and input assignment $\msfi$ are said to be \emph{co-legal}) if there exists an output assignment $\msfo$ such that $\gio\in \Pi$, in which case we say that $\msfo$ is a \emph{feasible solution} for $\gi$ (or simply for $G$ and $\msfi$).  For a DGP $\Pi$, we denote the set of legal input graphs by $\legpi=\{\gi\mid \exists\msfo:\gio\in \Pi\}$. 
	

	A \emph{distributed graph minimization problem (MinDGP)} (resp.,
	\emph{distributed graph maximization problem (MaxDGP)}) $\Psi$ is a pair
	$\pif $,
	where $\Pi$ is a DGP and
	$f:\Pi \rightarrow \mathbb{Z}$
	is a function, referred to as the \emph{objective function} of $\Psi$, that
	maps each IO graph
	$\gio\in \Pi$
	to an integer value
	$f(\gio)$.\footnote{%
		We assume for simplicity that the images of the objective functions used
		in the context of this paper are integral.
		Lifting this assumption and allowing for real numerical values would
		complicate some of the arguments, but it does not affect the validity of our
		results.}
	Given a co-legal graph $G$ and input assignment $\msfi$, define
	\[
	OPT_{\Psi}(G,\msfi) = \underset{\msfo:\gio\in \Pi}{\inf} \{f(\gio)\}
	\]
	if $\Psi$ is a MinDGP; and 
	\[
	OPT_{\Psi}(G, \msfi)= \underset{\msfo:\gio\in \Pi}{\sup} \{f(\gio\}
	\]
	if $\Psi$ is a MaxDGP.
	We often use the general term \emph{distributed graph optimization problem
		(OptDGP)} to refer to MinDGPs as well as MaxDGPs.
	Given an OptDGP $\Psi=\langle \Pi, f\rangle$ and co-legal graph $G$ and input assignment $\msfi$, the output assignment $\msfo$ is said to be an \emph{optimal solution} for $\gi$ (or simply for $G$ and $\msfi$) if $\msfo$ is a feasible solution for $\gi$ and $f(\gio)=OPT_\Psi(G,\msfi)$.

	\subsection{Proof Labeling Schemes}
	\label{section:pls}
	In this section we present the notion of proof labeling schemes as well as its approximation variants. To that end, we first present the notion of gap proof labeling schemes, as defined in \cite{EmekG20}.
	
	Fix some universe $\calu$ of configured graphs. A 
	\emph{gap proof labeling scheme (GPLS)} is a mechanism designed to
	distinguish the configured graphs in a \emph{yes-family}
	$\calf_{Y} \subset \calu$ from the configured graphs
	in a \emph{no-family}
	$\calf_{N} \subset \calu$, where
	$\calf_{Y} \cap \calf_{N} = \emptyset$. This is done by means of a (centralized) \emph{prover} and a (distributed)
	\emph{verifier} that play the following roles:
	Given a configured graph $G_s\in \calu$, if $G_s\in \calf_Y$,
	then the prover assigns a bit string $L(v)$, called the \emph{label} of $v$,
	to each node $v\in V$.
	Let $L^N(v)=\langle L(u_1),\dots ,L(u_{\deg(v)})\rangle$ be the vector of labels assigned to $v$'s neighbors. The verifier
	at node
	$v \in V$
	is provided with the $3$-tuple
	$\langle s(v), L(v), L^{N}(v) \rangle$
	and  returns a Boolean value $\varphi(v)$. 
	
	We say that the verifier \emph{accepts} $G_s$ if
	$\varphi(v) = \true$
	for all nodes
	$v \in V$;
	and that the verifier \emph{rejects} $G_s$ if
	$\varphi(v) = \false$
	for at least one node
	$v \in V$.
	The GPLS is said to be \emph{correct} if the following requirements hold for
	every configured graph
	$G_s \in \calu$:
	\begin{requirement}\label{gpls-r1}
		If $G_s\in \calf_Y$, then the prover produces a label assignment $L : \Vfunc$ such that the verifier accepts $G_s$.
	\end{requirement}
	\begin{requirement}\label{gpls-r2}
		If $G_s\in \calf_N$, then for any label assignment $L : \Vfunc$, the verifier rejects $G_s$.
	\end{requirement}
	We emphasize that no requirements are made for configured graphs
	$G_s \in \calu \setminus (\calf_Y \cup \calf_N)$;
	in particular, the verifier may either accept or reject these configured
	graphs (the same holds for configured graphs that do not belong to the
	universe $\calu$). The performance of a GPLS is measured by means of its \emph{proof size}
	defined to be the maximum length of a label $L(v)$ assigned by the prover to
	the nodes $v \in V$ assuming that 
	$G_s\in  \calf_Y$.
	\paragraph*{Proof Labeling Schemes for CGFs.}
	Consider some CGF $\Phi$ and let $\calu$ be the universe of all configured graphs. A \emph{proof labeling scheme (PLS)} for $\Phi$ is the GPLS over $\calu$ defined by setting the yes-family to be $\calf_{Y}=\Phi$; and the no-family to be $\calf_{N}=\calu\setminus\calf_{Y}$. In other words, a PLS for $\Phi$ determines whether a given configured graph $G_s$ belongs to $\Phi$.
	
	In this paper, we also define a relaxed model of PLSs for a CGF $\Phi$ in which we allow the verifier to accept configured graphs that are not "too far" from belonging to $\Phi$. To that end, we use the following distance measure which is widely used in the realm of property testing (see e.g., \cite{AlonKKR08}). 
	
	let $G_s$ and $G'_{s'}$ be two configured graphs. Given a parameter $\delta >0$, we say that $G_s$ and $G'_{s'}$ are \emph{$\delta$-close}  if $G'_{s'}$ is a configured subgraph of $G_s$ and $G'$ can be obtained from $G$ by removing at most $\delta m$ edges (or vice versa).

	Consider a CGF $\Phi$.
	We say that a configured graph $G_s$ is \emph{$\delta$-far} from belonging to $\Phi$ if $G'_{s'} \notin \Phi$ for any configured graph $G'_{s'}$ which is $\delta$-close to $G_s$. We define a \emph{$\delta$-testing proof labeling scheme ($\delta$-TPLS)} in the same way as a PLS for $\Phi$ with the sole difference that the no-family is defined by setting $\mathcal{F}_{N}=\{G_s\mid G_s \text{ is $\delta$-far from belonging to $\Phi$}\}$.
	\paragraph*{Proof Labeling Schemes for OptDGPs.}
	Consider some OptDGP $\Psi=\pif$ and let $\calu=\{\gio\mid \gi\in \legpi\}$. A \emph{proof labeling scheme (PLS)} for $\Psi$ is defined as a GPLS over $\calu$ by setting the yes-family to be \[\calf_Y=\{ \gio\in\Pi\mid f(\gio)=OPT_\Psi(G,\msfi)\}\] and the no-family to be $\calf_N=\calu\setminus \calf_Y$. In other words, a PLS for $\Psi$ determines for a given IO graph $\gio\in \calu$ whether the output assignment $\msfo:\Vfunc$ is an optimal solution (which means in particular that it is a feasible solution) for the co-legal graph $G=(V,E)$ and input assignment $\msfi:\Vfunc$.
	
	In the realm of OptDGPs, a relaxed model called approximate proof labeling scheme has been considered in \cite{aplspaper,EmekG20}. In this model, the correctness requirement of a PLS  are relaxed so that it may also accept feasible solutions that only approximate the optimal ones. Specifically, given an approximation parameter
	$\alpha \geq 1$,
	an \emph{$\alpha$-approximate proof labeling scheme ($\alpha$-APLS)} for an OptDGP
	$\Psi = \langle \Pi, f \rangle$
	is defined in the same way as a PLS for $\Psi$ with the sole difference that the no-family is defined by setting
	\[
	\mathcal{F}_{N}
	\, = \,
	\begin{cases}
		\mathcal{U} \setminus \left\{
		\gio\in\Pi 
		\mid f(\gio) \leq \alpha \cdot OPT_{\Psi}(G, \msfi)
		\right\}
		\, , &
		\text{if $\Psi$ is a MinDGP} \\
		\mathcal{U} \setminus \left\{
		\gio \in \Pi \mid f(\gio) \geq OPT_{\Psi}(G, \msfi) / \alpha
		\right\}
		\, , &
		\text{if $\Psi$ is a MaxDGP}
	\end{cases} \, .
	\]

	\sloppy

	\subsection{Locally Restricted Proof Labeling Schemes}
	\label{section:locally-restricted}
	In this paper, we focus on provers whose power is limited as follows.
	We say that a GPLS is \emph{locally restricted} if there exists a constant $c$ such that for every configured graph
	$G_{s} = \langle G = (V, E), s \rangle \in \mathcal{F}_{Y}$
	and for every node
	$v \in V$,
	the label $L(v)$ is computed by the prover as a function of $G_{s}(B^{r}(v))$,
	where
	$r = \log^{c} n$
	and $B^{r}(v)$ denotes the set of nodes at (hop) distance at most $r$ from $v$ in $G$.
	Equivalently, the prover is restricted to a distributed algorithm operating under the LOCAL model
	\cite{Linial92,Peleg2000}
	with polylogarithmic rounds. We emphasize that if $G_s\in \calf_{N}$, then the verifier is required to reject $G_s$ for any label assignment, including label assignments that were not produced in a locally restricted fashion.


	\section{Preliminaries}
	\label{section:preliminaries}
	
	\paragraph*{Sequentially Local Algorithms.}
	
	In the \emph{sequentially local (SLOCAL)} model, introduced in \cite{GhaffariKM17}, each node $v\in V$ maintains two (initially empty) bit strings denoted by $\text{info}(v)$ and $\text{decision}(v)$. Nodes are processed sequentially in an arbitrary order $p=v_1,\dots,v_n$ (i.e., irrespective of node ids). We refer to the time that node $v_i$ is processed as the $i$-th iteration of the algorithm. In the $i$-th iteration, $v_i$ has a read/write access to $\text{info}(u)$ for all nodes $u\in B^r(v_i)$, where $r\in \mathbb{Z}_{\geq 0}$ is a parameter referred to as the \emph{locality} of the algorithm. Following that, $v_i$ writes an irrevocable value into $\text{decision}(v_i)$ based strictly on $G_s(B^r(v_i))$ and the bit strings $\text{info}(u)$ of all $u\in B^r(v_i)$.
	
	A consequence of the seminal work of Ghafari et al.\
	\cite{GhaffariKM17,RozhonG20}
	is that any SLOCAL algorithm with
	$\log^{O (1)} n$
	locality can be simulated by a LOCAL algorithm with $\log ^{O(1)}n$ rounds. Therefore, in the context of a locally restricted GPLS, by allowing the prover to compute the label $L(v)$ of each node $v\in V$ using an SLOCAL algorithm with locality $r=\log ^{O(1)}n$ (rather than a LOCAL algorithm with $\log^{O(1)}n$ rounds), we do not increase the scheme's power.
	
	\sloppy
	\paragraph*{Comparison Schemes.}
	Let $\calu$ be a universe of configured graphs $G_{s_{a,b}}$, such that $G=(V,E)$ is a connected undirected graph and the configuration function $s_{a,b}:\Vfunc$ assigns two values $a(v),b(v)\in \mathbb{R}$ to each node $v\in V$. A \emph{comparison scheme} is a mechanism whose goal is to decide if $\sum_{v\in V}a(v)\geq \sum_{v\in V}b(v)$ for a given configured graph $G_{s_{a,b}}\in\calu$. Formally, a comparison scheme is defined as a GPLS over $\calu$ by setting the yes-family to be $\calf_{Y}=\{G_{s_{a,b}}\in \calu\mid \sum_{v\in V}a(v)\geq \sum_{v\in V}b(v)\}$; and the no family to be $\calf_{N}=\calu\setminus\calf_{Y}$. 
	
	In \cite[Lemma 4.4]{pls}, Korman et al.\ present a generic design for
	comparison schemes as follows. Consider a configured graph
	$G_{s_{a,b}}\in\calu$. The label assignment $L:\Vfunc$ constructed by the
	prover encodes a spanning tree of $G$ rooted at some (arbitrary) node $r\in V$
	(see \cite[Lemma 2.2]{pls} for details on spanning tree construction). In
	addition, the prover encodes the sum of $a(\cdot)$ and $b(\cdot)$ values in
	the sub-tree rooted at node $v$ for each $v\in V$. This allows the verifier to
	check that the sums assigned at each node $v\in V$ are correct (using the sums
	assigned to $v$'s children); and the verifier at the root $r$ can evaluate if
	$G_{s_{a,b}}\in\calf_{Y}$. The proof size of this scheme is $O(\log
	n+M_{a,b})$, where $M_{a,b}$ is the maximum length (in bits) of values
	$\sum_{v\in U}a(v)$ and $\sum_{v\in U}b(v)$ over all node-subsets $U\subseteq
	V$. This comparison scheme construction is used as an auxiliary tool in the compilers presented in Sections \ref{section:compiler-optdgps} and \ref{section:compiler-cgfs}.  
	
	\sloppy

	\section{Compiler for OptDGPs}
	\label{section:compiler-optdgps}
	In this section, we present our generic compiler for OptDGPs. It is divided into five subsections as follows. First, in Section \ref{section:optdgp-charachterization} we characterize the OptDGPs that are suited for our compiler, referred to as canonical OptDGPs, based on the notions of locally checkable labelings and covering/packing OptDGPs (these terms are formally defined in Section \ref{section:optdgp-charachterization}). In Section \ref{section:properties-of-canonical-optdgps}, we establish an important property of optimal solutions for covering/packing OptDGPs that serve the compiler construction. Sections \ref{section:part-one-partition} and \ref{section:part-two-label-and-verification} are dedicated to the compiler construction. More formally, these sections constructively prove the following theorem.
	\begin{theorem}\label{theorem:main-optdgp}
		Let $\Psi$ be a canonical OptDGP that admits an $\alpha$-APLS with a proof size of $\ell_{\Psi,\alpha}$. For any constant $\epsilon>0$, there exists a locally restricted $(\alpha(1+\epsilon))$-APLS for $\Psi$ with a proof size of $\ell_{\Psi,\alpha}+O(\log n)$.
	\end{theorem}
	
	For convenience, the compiler construction is divided between Sections \ref{section:part-one-partition} and \ref{section:part-two-label-and-verification} as follows. In Section \ref{section:part-one-partition}, we present an SLOCAL algorithm with logarithmic locality that partitions the nodes into disjoint clusters, such that the subgraph induced by each cluster is of logarithmic diameter. The goal of this partition is to enable the prover to construct the label of a node as a function of the subgraph induced by its cluster (and possibly some nodes adjacent to its cluster) without information on nodes that are farther away. This partition facilitates the label assignment and verification process described in Section \ref{section:part-two-label-and-verification}. In Section \ref{section:high-level}, we provide a high-level overview of the SLOCAL algorithm and how it is used in the label assignment and verification process.

	\subsection{Canonical OptDGPs}
	\label{section:optdgp-charachterization}
	\paragraph*{Locally Checkable Labelings.}
	A DGP $\Pi$ is said to be a \emph{locally checkable labeling (LCL)} (cf.\ \cite{NaorS95}) if there exists a Boolean predicate family $\mathcal{L}^{\Pi}=\{
	p_{d, \ell}^{\Pi} : (\{ 0, 1 \}^{*})^{d + 1} \rightarrow \{ \true, \false \}
	\}_{d \in \mathbb{Z}_{\geq 0}, \ell \in \{ 0, 1 \}^{*}}$ such that for every legal input graph
	$\gi\in \legpi$,
	an output assignment
	$\msfo:\Vfunc$
	is a feasible solution for $G$ and $\msfi$ if and only if
	$p_{\deg(v), \msfi(v)}^{\Pi}(\msfo(v), \msfo(u_{1}), \dots, \msfo(u_{\deg(v)})) = \true$
	for every node
	$v \in V$
	with neighbors
	$u_{1}, \dots, u_{\deg(v)}$. 
	
	For convenience, we assume that the local input $\msfi(v)$ of a node $v \in V$ is partitioned into two fields, denoted by $\prdiv$ and $\dataiv$,
	where the former (fully) determines the predicate $p^{\Pi}_{\deg(v), i(v)}$ associated with $\deg(v)$ and $\msfi(v)$ and the latter encodes all other pieces of information included in $\msfi(v)$.
	This allows us to slightly abuse the notation and write $p^{\Pi}_{\prdiv}$ instead of $p^{\Pi}_{\deg(v), \msfi(v)}$.
	We further assume that the Boolean predicate family $\mathcal{L}^{\Pi}$ includes the trivial tautology predicate
	$taut_d:(\{ 0, 1 \}^{*})^{d + 1} \rightarrow \{ \true, \false \}$,
	$d \in \mathbb{Z}_{\geq 0}$,
	that satisfies
	$taut_{d}(x) = \true$
	for every
	$x \in (\{ 0, 1 \}^{*})^{d + 1}$
	and that this predicate is encoded by writing the designated bit string $taut_{d}$ in the $\mathtt{prd}(\cdot)$ field of the local input.
	
	We say that an LCL $\Pi$ is \emph{self-induced} if the following two conditions are satisfied for every legal input graph $\gi\in \legpi$ and node subset $U\subseteq V$: (1) $\gi(U)\in \legpi$; and (2) if $\msfi':\Vfunc$ is the input assignment derived from $\msfi$ by setting $\mathtt{data}(\msfi'(v))=\dataiv$ and $\mathtt{prd}(\msfi'(v))=taut_{\deg(v)}$ for every $v\in U$, then $G_{\msfi'}\in\legpi$. 
	
	Let $\Pi$ be an LCL and let $\gi\in \legpi$. For a subset $U\subseteq V$ of nodes, we denote by $\inner(U)=\{u\in U\mid N_G(u)\subseteq U\}$ the set of nodes in $U$ for which every neighbor is in $U$ and define $\inner^2(U)=\inner(\inner(U))$ and $\rim(U)=U\setminus\inner^2 (U)$. We say that a function $g:U\rightarrow\{0,1\}^*$ \emph{respects} $\Pi$ if $p^{\Pi}_{\mathtt{prd}(\msfi(v))}(g(v),g(u_1),\dots,g(u_{\deg(v)}))=\true$ for each $v\in \inner(U)$ with neighbors $u_1,\dots, u_{\deg(v)}$.
	
	\paragraph*{Canonical OptDGPs.}
	A MinDGP (resp., MaxDGP) $\Psi=\pif$ is said to be a \emph{covering} (resp., \emph{packing}) OptDGP  if the following conditions hold: (1) $\Pi$ is an LCL; (2) for each $n$-node IO graph $\gio\in \Pi$, there exists a positive integer $k=k(\Pi,n)=n^{O(1)}$ such that the output assignment $\msfo$ assigns a nonnegative integer $\msfo(v)\in\{0,\dots,k\}$, referred to as $v$'s \emph{multiplicity}, to each node $v\in V$; (3) for each predicate $p^{\Pi}_{d,\ell}\in\mathcal{L}^{\Pi}$, if $p^{\Pi}_{d,\ell}(x_0, x_1,\dots,x_{d})=\true$ for nonnegative integers $x_0,x_1,\dots,x_{d}\in \{0,\dots,k\}$, then $p^{\Pi}_{d,\ell}(x'_0, x'_1,\dots,x'_{d})=\true$ for any nonnegative integers $x'_0,x'_1,\dots,x'_{d}\in \{0,\dots,k\}$ that satisfy $x'_j\geq x_j$ (resp., $x'_j\leq x_j$) for all $0\leq j\leq d$; (4) for every legal input graph $\gi\in \legpi$, there exists a node-weight function $w:V\rightarrow\{1,\dots n^{O(1)}\}$ such that the weight $w(v)$ of node $v$ is encoded in $v$'s local input field $\dataiv$; and (5) $f(\gio)=\sum_{v\in V}w(v)\cdot \msfo(v)$ for every $\gio\in\Pi$. The OptDGP $\Psi=\pif$ is said to be \emph{canonical} if it is covering/packing and $\Pi$ is self-induced. 
	
	Consider a covering/packing OptDGP $\Psi=\pif$. Let $\gi\in \legpi$ be a legal input graph with the underlying node-weight function $w:V\rightarrow\{1,\dots, n^{O(1)}\}$ and let $U\subseteq V$. Given a function $g:U\rightarrow\{0,\dots,k(\Pi,n)\}$ that assigns a multiplicity value $g(u)$ to each node $u\in U$, we define $w(U,g)=\sum_{u\in U}w(u)\cdot g(u)$. 
	
	For a covering MinDGP $\Psi$, let $w_{\min}(U)$ denote the minimum possible value of $w(U,g)$ obtained by a function $g: U \rightarrow \{ 0, \dots, k(\Pi, n) \}$ that respects $\Pi$. Let $N^{2}(U)$ be the set of nodes in $V\setminus U$ at distance at most $2$ from a node in $U$. For a packing MaxDGP, let $w_{\max}(U)$ denote the maximum possible value of $w(U,g)$ obtained by a function $g: U\cup N^{2}(U) \rightarrow \{ 0, \dots, k(\Pi, n) \}$ that satisfies (1) $g(v)=0$ for each node $v\in N^{2}(U)$; and (2) $g$ respects $\Pi$.

	\subsection{Properties of Optimal Solutions for Covering/Packing OptDGPs} 
	\label{section:properties-of-canonical-optdgps}
	In the following lemmas, we establish important properties regarding optimal solutions of covering and packing OptDGPs. Consider a covering (resp., packing) OptDGP $\Psi=\pif$. Let $\gio\in \Pi$ be an IO graph such that $\msfo:V\rightarrow\{0,\dots,k(\Pi,n)\}$ is an optimal solution for $G$ and $\msfi$. 
	
	\begin{lemma}\label{lemma:min-dgp-inner}
		If $\Psi$ is a covering MinDGP, then $w(\inner^2(U),\msfo)\leq w_{\min}(U)$ for any $U\subseteq V$.
	\end{lemma}
	\begin{proof}
		Assume by contradiction that $w(\inner^2(U),\msfo)> w_{\min}(U)$ for some $U\subseteq V$. This means that there exists an assignment $\msfo':U\rightarrow\{0,\dots,k(\Pi,n)\}$ that respects $\Pi$, such that $w(\inner^2(U),\msfo)> w(U,\msfo')$. Let $\tilde{\msfo}$ be the output assignment defined as follows: $\tilde{\msfo}(v)=\msfo(v)$ for all $v\in V\setminus U$; $\tilde{\msfo}(v)=\msfo'(v)$ for all $v\in \inner^2(U)$; and $\tilde{\msfo}(v)=\max\{\msfo(v),\msfo'(v)\}$ for all $v\in \rim(U)$. Recall that $\msfo$ is a feasible solution for $G$ and $\msfi$ and that $\msfo'$ respects $\Pi$. Since $\Psi$ is a covering OptDGP, we get that $p^{\Pi}_{\prdiv}(\tilde{\msfo}(v),\tilde{\msfo}(u_1),\dots,\tilde{\msfo}(u_{\deg_{G}(v)}))=\true$ for each node $v\in V$ (where  $u_1,\dots,u_{\deg_G(v)}$ denote $v$'s neighbors in $G$). It follows that $\tilde{\msfo}$ is a feasible solution with objective value $f(G_{\msfi,\tilde{\msfo}})=w(V,\tilde{\msfo})\leq w(V\setminus U, \msfo)+w(\inner^2(U),\msfo')+w(\rim(U),\msfo)+w(\rim(U),\msfo')= w(V,\msfo)-w(\inner^2(U),\msfo)+w(U,\msfo')<w(V,\msfo)=f(\gio)$ which contradicts the optimality of $\msfo$.
	\end{proof}
	
	\begin{lemma}\label{lemma:max-dgp-inner}
		If $\Psi$ is a packing MaxDGP, then $w(U,\msfo)\geq w_{\max}(\inner^{2}(U))$ for any $U\subseteq V$.
	\end{lemma}
	\begin{proof}
		Assume by contradiction that $w(U,\msfo)< w_{\max}(\inner ^{2}(U))$ for some $U\subseteq V$. This means that there exists an assignment $\msfo':U\rightarrow\{0,\dots,k(\Pi,n)\}$ that respects $\Pi$, such that $\msfo'(v)=0$ for each $v\in \rim(U)$, and $w(U,\msfo)< w(\inner^2(U),\msfo')$. Let $\tilde{\msfo}$ be the output assignment defined as follows: $\tilde{\msfo}(v)=\msfo(v)$ for all $v\in V\setminus U$; $\tilde{\msfo}(v)=\msfo'(v)$ for all $v\in \inner^2(U)$; and $\tilde{\msfo}(v)=0$ for all $v\in \rim(U)$. Recall that $\msfo$ is a feasible solution for $G$ and $\msfi$ and that $\msfo'$ respects $\Pi$. Since $\Psi$ is a packing OptDGP, we get that $p^{\Pi}_{\prdiv}(\tilde{\msfo}(v),\tilde{\msfo}(u_1),\dots,\tilde{\msfo}(u_{\deg_{G}(v)}))=\true$ for each node $v\in V$ (where  $u_1,\dots,u_{\deg_{G}(v) }$ denote $v$'s neighbors in $G$). It follows that $\tilde{\msfo}$ is a feasible solution with objective value $f(G_{\msfi,\tilde{\msfo}})=w(V,\tilde{\msfo})=w(V,\msfo)-w(U,\msfo)+w(\inner^2(U),\msfo')>w(V,\msfo)=f(\gio)$ which contradicts the optimality of $\msfo$.
	\end{proof}

	\subsection{Overview} 
	\label{section:high-level}
	In Section \ref{section:part-one-partition}, we present an SLOCAL algorithm called $\partopt$ that partitions the nodes of a given IO graph $\gio$ into clusters. Before formally describing the algorithm in Section \ref{section:part-one-partition}, let us provide some intuition by presenting the high-level idea of the partition and how it is used in the label and verification process (as described in Section \ref{section:part-two-label-and-verification}) for the case of a canonical MinDGP $\Psi$ (the high-level idea for MaxDGPs is similar and the differences are mostly technical).  
	
	Given an IO graph $\gio$, where $\msfo$ is an optimal solution, we use a ball growing argument to obtain a partition of the nodes into clusters such that: (1) the subgraph induced by each cluster is of logarithmic diameter; and (2) the total weight of nodes in the rim of clusters (i.e., nodes with distance at most $2$ from a different cluster) is an $\epsilon$-fraction of the total weight of inner nodes of clusters (i.e., nodes with distance at least $3$ from a different cluster). 
	
	The goal of the partition obtained by $\partopt$ is to allow the prover to compute the label assigned to each node based on its cluster. Essentially, the prover seeks to provide the verifier with a proof that the partition satisfies $2$ main properties: (1) for each cluster $V_j$, the weight of the given solution induced on the inner nodes is at most the weight of an approximately optimal (global) solution induced on the cluster; and (2) the total weight of nodes in the rim of clusters is at most an $\epsilon$-fraction of the total weight of inner nodes. 
	
	While providing proof for the first property is rather straightforward using the labels of an $\alpha$-APLS for $\Psi$ in a black box manner, providing proof for the second property is somewhat more challenging. The reason is that the property as presented above is rather global --- not every cluster is guaranteed to have at most an $\epsilon$-fraction of its weight assigned to the rim nodes. Constructing a label that sums the total weights of rim and inner nodes of all clusters is a global task and can not be accomplished in a locally restricted fashion. To that end, during $\partopt$, we may assign some nodes with a \emph{secondary affiliation} to an adjacent cluster. The idea is that for each cluster $V_j$, the sum between weights of nodes with secondary affiliation to $V_j$ and nodes in $\rim(V_j)$ that do not have a secondary affiliation to any cluster is bounded by an $\epsilon$-fraction of $V_j$'s inner nodes weight. 
	
	Throughout $\partopt$, each node $v$ maintains a \emph{color} whose role is to keep track of  the changeability status of $v$'s secondary affiliation. The color white indicates that the secondary affiliation may still change; whereas black indicates that the secondary affiliation is final.    
	
	\subsection{Partition Algorithm} 
	\label{section:part-one-partition}
	
	\paragraph*{Algorithm's Description.}
	
	We now provide a formal description of the $\partopt$ algorithm. Consider a canonical OptDGP $\Psi=\pif$. Let $\gio\in \Pi$ be an IO graph such that $\msfo$ is an optimal solution for $G$ and $\msfi$. The algorithm partitions the nodes of $G$ into (possibly empty) clusters $V=V_1\dot{\cup}\dots\dot{\cup}V_n$. As usual in the SLOCAL model, the nodes are processed sequentially in $n$ iterations based on an arbitrary order $v_1,\dots,v_n$ on the nodes, where node $v_j$ is processed in the $j$-th iteration.
	
	Throughout the execution of $\partopt$, each node $v\in V$ maintains three fields referred to as $\text{cluster}(v)$, $\text{sec}(v)$, and $\text{color}(v)$. The field $\text{cluster}(v)$ is initially empty and its role is to identify $v$'s cluster, where each cluster $V_j$ is identified by the id of node $v_j$ which is processed in the $j$-th iteration, i.e., $V_j=\{v\mid \text{cluster}(v)=id(v_j)\}$. The field $\text{sec}(v)$ is initially empty and its role is to identify $v$'s \emph{secondary affiliation} to a cluster if such affiliation exists (otherwise it remains empty throughout the algorithm). The field $\text{color}(v)\in \{\text{black},\text{white}\}$, initially set to white, maintains $v$'s \emph{color}.

	We describe the $j$-th iteration of $\partopt$ as follows. Let $G_j$ be the subgraph induced on $G$ by $V\setminus(V_1\cup\dots\cup V_{j-1})$. If $v_j\in V_1\cup\dots\cup V_{j-1}$, then we define $V_j=\emptyset$ and finish the iteration; so, assume that $v_j$ is a node in $G_j$. For an integer $r\in \mathbb{Z}_{\geq 0}$, let $D_j^r$ be the set of nodes at distance exactly $r$ from $v_j$ in $G_j$ and let $B^r_j=\bigcup_{r'=0}^{r}D_j^{r'}$. Let $white_j$ be the set of nodes in $G_j$ that are colored white in the beginning of the $j$-th iteration.   
	
	Suppose that $\Psi$ is a MinDGP. We define $r(j)$ to be the smallest integer that satisfies $w(\inner^2(B_j^{r(j)+6}),\msfo)\leq (1+\epsilon) \cdot w(\inner ^2(B_j^{r(j)+2}),\msfo)$. Notice that $\inner^2(\cdot)$ is taken with respect to nodes in $G$ (and not $G_j$), i.e., $\inner^2(B_j^{r(j)+6})$ (resp., $\inner ^2(B_j^{r(j)+2})$) is the set of nodes in $B_j^{r(j)+6}$ (resp., $B_j^{r(j)+2}$) for which every node within distance $2$ in $G$ is in $B_j^{r(j)+6}$ (resp., $B_j^{r(j)+2}$). In the case that $\Psi$ is a MaxDGP, define $r(j)$ to be the smallest integer that satisfies $w(B_j^{r(j)+6},\msfo)\leq (1+\epsilon) \cdot w(B_j^{r(j)+2},\msfo)$.

	Following the computation of $r(j)$, we define the cluster $V_j$ and modify the color and secondary affiliation of some nodes as follows (this process is the same for MinDGPs and MaxDGPs). Let $X_j$ be the set of white nodes in $\inner^2(B_j^{r(j)+6})$ at distance exactly $r(j)+3$ from $v_j$, and let $Y_j$ be the set of white nodes in $\inner^2(B_j^{r(j)+6})$ at distance exactly $r(j)+4$ from $v_j$ that have a neighbor in $X_j$. We complete the $j$-th iteration by setting $\text{cluster}(v)=id(v_j)$ for each node $v\in B_j^{r(j)+2}$ (i.e., setting $V_j=B_j^{r(j)+2}$); $\text{sec}(v)=id(v_j)$ for each node $v\in X_j\cup Y_j$; and  $\text{color}(v)=\text{black}$ for each node $v\in X_j$.
	\sloppy 
	
	\paragraph*{Algorithm's Properties.}
	We go on to analyze some properties of $\partopt$. Consider a cluster $V_j$. Let $\mathtt{sec}(V_j)=\{v\mid \text{sec}(v)=id(v_j)\}$ be the set of nodes whose secondary affiliation is to $V_j$ by the end of the algorithm and let $\mathtt{ext}(V_j)=V_j\cup\mathtt{sec}(V_j)$.
	
	\begin{lemma}\label{lemma:cluter-diameter}
		The subgraphs $G(V_j)$ and $G(\mathtt{ext}(V_j))$ induced on $G$ by $V_j$ and $\mathtt{ext}(V_j)$, respectively,  are connected and have diameter $O(\log n)$ for each $j\in[n]$.
	\end{lemma}
	\begin{proof}
		Suppose that $V_j\neq \emptyset$ (as the lemma is trivial otherwise). First, observe that by definition, all nodes $v\in V_j$ are reachable from $v_j$ in $G(V_j)$, thus $G(V_j)$ is connected.
		
		To see that $G(\mathtt{ext}(V_j))$ is connected, we first observe that the subgraph $G(V_j\cup X_j\cup Y_j)$ is connected. By the time cluster $V_j$ is determined, we color the nodes of $X_j$ black. Thus, their secondary affiliation remains to $V_j$ throughout the algorithm. At termination, it follows that $\mathtt{ext}(V_j)=V_j\cup X_j\cup Y$ for some $Y\subseteq Y_j$. Since the nodes of $Y$ all have a neighbor in $X_j$, we get that $G(\mathtt{ext}(V_j))= G(V_j\cup X_j\cup Y)$ is connected. 
		
		To show that the diameters of $G(V_j)$ and $G(\mathtt{ext}(V_j))$ are $O(\log n)$ it is sufficient to show that $r(j)=O(\log n)$. We use a ball growing argument. By definition, for every $r'<r(j)+6$, it holds that
		$$w(\inner^2(B_j^{r(j)+6}),\msfo)\geq w (\inner^2(B_j^{r'}),\msfo)>(1+\epsilon)\cdot w (\inner^2(B_j^{r'-4}),\msfo)>$$
		$$(1+\epsilon)^2\cdot w (\inner^2(B_j^{r'-8}),\msfo)>\dots$$
		if $\Psi$ is a MinDGP; and
		$$w(B_j^{r(j)+6},\msfo)\geq w (B_j^{r'},\msfo)>(1+\epsilon)\cdot w (B_j^{r'-4},\msfo)>(1+\epsilon)^2\cdot w (B_j^{r'-8},\msfo)>\dots$$
		if $\Psi$ is a MaxDGP. Since the terms $w(\inner^2(B_j^{r(j)+6}),\msfo)$ and $w(B_j^{r(j)+6},\msfo)$ are both bounded by a polynomial of $n$, it follows that $r(j)=O((1/\epsilon)\log n)=O(\log n)$ in both cases.
	\end{proof}
	
	A simple observation derived from Lemma \ref{lemma:cluter-diameter} is that $\partopt$ has locality $O(\log n)$. This observation combined with the results of \cite{GhaffariKM17,RozhonG20} lead to the following corollary.
	\begin{corollary}\label{corollary:round-complexity}
		The algorithm $\partopt$ can be simulated by a LOCAL algorithm with polylogarithmic round-complexity.
	\end{corollary}
	
	For each $j\in [n]$, define $S_j=\mathtt{sec}(V_j)\cup \{v\in \rim(V_j)\mid \text{sec}(v)\text{ is empty}\}$ as the set of nodes composed of nodes outside of $V_j$ whose secondary affiliation is to $V_j$ and nodes in $\rim(V_j)$ that do not have a secondary affiliation. The following observation establishes an important property regarding the sets $\rim(V_j)$ and $S_j$.
	
	\begin{observation}\label{observation:rim}
		$\bigcup_{j\in [n]}\rim(V_j)\subseteq \bigcup_{j\in [n]}S_j$
	\end{observation}
	
	\begin{proof}
		Consider a node $v\in\rim(V_j)$ for some $j\in [n]$. By definition, if $\text{sec}(v)$ is empty, then $v\in S_j$. If $\text{sec}(v)$ is not empty, then there exists some $j'\in[n]$ such that $v\in \mathtt{sec}(V_{j'})$, and therefore $v\in S_{j'}$. Overall, we get that $v\in \bigcup_{\ell\in [n]}S_\ell$.
	\end{proof}

	\subsection{Labels and Verification} 
	\label{section:part-two-label-and-verification}
	In this section, we describe the label assignment and verification process of our compiler. We start in Section \ref{section:label-mindgp}, by describing the compiler's label and verification for MinDGPs. In Section \ref{section:label-maxdgp}, we go on to describe the changes required to establish the same for MaxDGPs. In both cases, we establish the proof size and correctness of our construction, thus proving Theorem \ref{theorem:main-optdgp}. 
	\subsubsection{MinDGPs}
	\label{section:label-mindgp}
	Consider a canonical MinDGP $\Psi=\pif$ and an IO graph $\gio\in \Pi$, where $\msfo$ is an optimal solution for $G$ and $\msfi$. The prover uses the SLOCAL algorithm $\partopt$ presented in Section \ref{section:part-one-partition} to compute the values $r(j)$ and subsets $V_j$, $\mathtt{sec}(V_j)$, $\mathtt{ext}(V_j)=V_j\cup \mathtt{sec}(V_j)$, and $S_j=\mathtt{sec}(V_j)\cup \{v\in \rim(V_j)\mid \text{sec}(v)\text{ is empty}\}$ for all $j\in[ n]$.
	
	The goal of the prover is to provide proof of four properties satisfied by the given solution $\msfo$ and the outcome of $\partopt$. We refer to those properties as \emph{feasibility}, \emph{rim}, \emph{growth}, and \emph{optimality}. The four properties are defined as follows. The feasibility property states that $\msfo$ is a feasible solution for $G$ and $\msfi$, i.e., $\gio\in\Pi$; the rim property states that for each $j\in [n]$ and node $v\in\rim(V_j)$, there exists $j'\in [n]$, such that $v\in S_{j'}$; the growth property states that $w(S_j,\msfo)\leq \epsilon \cdot w(\inner ^2(V_j),\msfo)$ for each $j\in[n]$; and the optimality property states that $w(\inner ^2(V_j),\msfo)\leq \alpha \cdot w_{\min}(V_j)$ for each $j\in [n]$. 
	
	The prover provides its proof by means of a label assignment $L:\Vfunc$ that assigns each node $v$ with a label $L(v)=\langle L_{\text{feas}}(v),L_{\text{rim}}(v),L_{\text{grw}}(v),L_{\text{opt}}(v)\rangle$. The label $L(v)$ is composed of the fields $L_{\text{feas}}(v)$, $L_{\text{rim}}(v)$, $L_{\text{grw}}(v)$, and $L_{\text{opt}}(v)$ that provide proof for the feasibility, rim, growth, and optimality properties, respectively. 
	
	The field $L_{\text{feas}}(\cdot)$ provides a proof for the feasibility property by setting $L_{\text{feas}}(v)=\msfo(v)$ for each node $v\in V$. Notice that since $\Pi$ is an LCL, verifying $\msfo$'s feasibility is done by checking that $L_{\text{feas}}(v)=\msfo(v)$, and $p_{\prdiv}^{\Pi}(L_{\text{feas}}(v), L_{\text{feas}}(u_1),\dots L_{\text{feas}}(u_{\deg_G(v)}))=\true$ at each node $v$ with neighbors $u_1,\dots,u_{\deg_G(v)}$.
	
	The field $L_{\text{rim}}(\cdot)$ provides a proof for the rim property as follows. First, the sets $V_j$ and $S_j$ are encoded for all $1\leq j\leq n$, where each of the sets is identified by $id(v_j)$. In addition, each node $v\in \rim(V_j)$ is assigned the minimal distance to a node $u\notin V_j$ (notice that by definition, these values are either $1$ or $2$). This allows the verifier to check that for each node $v\in \rim(V_j)$, there exists $j'\in [n]$ such that $v\in S_{j'}$, i.e., verify that the rim property is satisfied.

	The field $L_{\text{grw}}(\cdot)$ provides a proof for the growth property simply by using a comparison scheme (as defined in Section \ref{section:preliminaries}) that compares between $w(S_j,\msfo)$ and $\epsilon \cdot w(\inner ^2(V_j),\msfo)$. This comparison scheme is used concurrently for each $\mathtt{ext}(V_j)\neq \emptyset$, based on a shortest paths tree of $G(\mathtt{ext}(V_j))$ rooted at node $v_j$. Observe that by Lemma \ref{lemma:cluter-diameter}, this tree spans the nodes of $\mathtt{ext}(V_j)$ and has diameter $O(\log n)$. 
	
	The field $L_{\text{opt}}(\cdot)$ provides a proof for the optimality property as follows. First, for each $V_j\neq \emptyset$, the prover computes an assignment $g_j:V_j\rightarrow \{0,\dots ,k(\Pi,n)\}$, such that $g_j$ respects $\Pi$ and $w(V_j,g_j)=w_{\min}(V_j)$. The prover assigns each node $v\in V_j$ with the multiplicity $g_j(v)$ and proves that $ w(\inner^2(V_j),\msfo)\leq  w(V_j,g_j)$ by means of a comparison scheme based on a shortest paths (spanning) tree of $G(V_j)$ rooted at node $v_j$. Finally, the prover proves that $w(V_j,g_j)\leq \alpha\cdot w_{\min}(V_j)$ by means of an $\alpha$-APLS for $\Psi$ on the configured subgraph $\gio(V_j)$. Notice that an $\alpha$-APLS for $\Psi$ is well-defined over the instance $\gio(V_j)$ since $\Pi$ is self-induced, and thus $\gi(V_j)\in\legpi$.

	\paragraph*{Proof Size and Correctness.}
	
	We observe that the label assignment produced by the prover can be computed by means of an SLOCAL algorithm with locality $O(\log n)$ and thus it can be simulated by a locally restricted prover. Moreover, for each node $v\in V$, the sub-labels  $L_{\text{feas}}(v)$, $L_{\text{rim}}(v)$, and $L_{\text{grw} }(v)$ are of size $O(\log n)$; whereas $L_{\text{opt}}(v)$ is of size $\ell_{\Psi,\alpha}+O(\log n )$, where $\ell_{\Psi,\alpha}$ is the proof size of an $\alpha$-APLS for $\Psi$. Overall, the proof size of this scheme is $\ell_{\Psi,\alpha}+O(\log n)$. 
	
	Regarding the correctness requirements, we start by showing the completeness requirement, i.e., we show that if $\msfo$ is an optimal solution for $G$ and $\msfi$, then the verifier accepts $\gio$. To that end, it is sufficient to show that all four aforementioned properties are satisfied. The feasibility property holds since by definition, $\msfo$ is a feasible solution for $G$ and $\msfi$; the rim property follows directly from Observation \ref{observation:rim}; the growth property holds by the construction of the clusters $V_j$; and the optimality property follows from Lemma \ref{lemma:min-dgp-inner}. We note that as established in Lemma \ref{lemma:min-dgp-inner}, the optimality property is satisfied by $\msfo$ with parameter $\alpha=1$. However, providing proof for this stronger property might be costly in terms of proof size. Thus, to obtain a small proof size, we settle for an approximated version.
	
	As for the soundness requirement, consider an IO graph $\gio$ such that the verifier accepts $\gio$. This means that all four properties hold for $\gio$. First, observe that by the feasibility property, it holds that $\gio \in\Pi$. Let $V^{L}_1,\dots V^{L}_k$ and $S^{L}_1,\dots S^{L}_k$ be the subsets $V_j$ and $S_j$ encoded by the prover in the field $L_{\text{rim}}(\cdot)$. By the rim property, it holds that $\bigcup_{j\in[k]}\rim(V^{L}_j)\subseteq \bigcup_{j\in[k]}S^{L}_j$. From the growth property it follows that $\epsilon\cdot w(\bigcup _{j\in[k]}\inner ^2(V^{L}_j),\msfo)\geq w(\bigcup_{j\in[k]}S^{L}_{j},\msfo)\geq w(\bigcup_{j\in[k]}\rim(V^{L}_j),\msfo)$.  Let $\msfo^{*}:V\rightarrow\{0,\dots ,k(\Pi,n)\}$ be an optimal solution for $G$ and $\msfi$. We observe that for any $U\subseteq V$, the assignment of $\msfo^{*}$ on the nodes of $U$ must respect $\Pi$, and therefore $w(U,\msfo^{*})\geq w_{\min}(U)$. The optimality property combined with the last observation implies that $w(\bigcup_{j\in[k]}\inner ^2(V^{L}_j),\msfo)\leq \alpha (w_{\min}(V^{L}_1)+\dots+w_{\min}(V^{L}_k))\leq \alpha(w(V^{L}_1,\msfo^{*})+\dots +w(V^{L}_k,\msfo^{*}))= \alpha \cdot w(V,\msfo^{*})=\alpha \cdot f(G_{\msfi,\msfo^{*}})$. Combining this inequality with the rim and growth properties implies that $f(\gio)=w(V,\msfo)=w(\bigcup_{j\in[k]}\inner ^2(V^{L}_j),\msfo)+w(\bigcup_{j\in[k]}\rim(V^{L}_j),\msfo) \leq (1+\epsilon)\cdot w(\bigcup_{j\in[k]}\inner ^2(V^{L}_j),\msfo)\leq \alpha\cdot(1+\epsilon) \cdot f(G_{\msfi,\msfo^{*}})$, thus establishing the soundness requirement.

	\subsubsection{MaxDGPs}
	\label{section:label-maxdgp}
	Consider a canonical MaxDGP $\Psi=\pif$ and an IO graph $\gio\in \Pi$, where $\msfo$ is an optimal solution for $G$ and $\msfi$. We apply the same label construction and verification process  as the one described in Section \ref{section:label-mindgp} for MinDGPs with the only differences being in the definitions of the growth and optimality properties, along with their designated label fields $L_{\text{grw} }(\cdot)$ and $L_{\text{opt}}(\cdot)$. Let $V_j$, $\mathtt{sec}(V_j)$, $\mathtt{ext}(V_j)$, and $S_j$ denote the subsets obtained by the $\partopt$ algorithm (as defined in Section \ref{section:part-one-partition}) for all $1\leq j\leq n$. Let $T_j=\inner^{2}(V_j)\cup S_j$ and let $g_j$ be an assignment $g_j:T_j\cup N^{2}(T_j)\rightarrow \{0,\dots,k(\Pi,n)\}$ that realizes $w_{\max}(T_j)$.
	
	The growth property for MaxDGPs states that $(1+\epsilon)\cdot w(V_j,\msfo)\geq w(T_j,g_j)$. The field $L_{\text{grw}}(\cdot)$ provides a proof for the growth property as follows. First, the multiplicity $g_j(v)$ is assigned to each $v\in T_j$. Then, similarly to the construction in the MinDGP case, the growth property can be proven by means of a comparison scheme that compares between  $w(T_j,g_j)$ and $(1+\epsilon)\cdot w(V_j,\msfo)$. Notice that the subsets $T_j$ can be deduced from the field $L_{\text{rim}}(\cdot)$.
	
	The optimality property states that $\alpha\cdot w(T_j,g_j)\geq w_{\max}(T_j)$. The field $L_{\text{opt}}(\cdot)$ provides a proof for the optimality property by means of an $\alpha$-APLS for $\Psi$ that evaluates the proposed solution $g_j(v)$ on the configured subgraph $\gio(T_j)$.

	\paragraph*{Proof Size and Correctness.}
	Similarly to the case with MinDGPs, the APLS described is locally restricted and has a proof size of  $\ell_{\Psi,\alpha}+O(\log n)$. In terms of correctness, we first analyze the completeness requirement. Consider a given IO graph $\gio$ such that $\msfo$ is an optimal solution for $G$ and $\msfi$. The feasibility and rim properties hold as established in \ref{section:label-mindgp}. The optimality property holds trivially. As for the growth property, first observe that the cluster $V_j$ is defined so that $(1+\epsilon)\cdot w(V_j,\msfo)\geq w(B_{j}^{r(j)+6},\msfo)$. Since $T_j\subseteq \inner ^{2}(B_{j}^{r(j)+6})$, it follows from Lemma \ref{lemma:max-dgp-inner} that $(1+\epsilon)\cdot w(V_j,\msfo)\geq w(B_{j}^{r(j)+6},\msfo)\geq w_{\max}(\inner^{2}(B_{j}^{r(j)+6}))\geq  w(T_j,g_j)$, which establishes the growth property.
	
	Regarding the soundness, consider an IO graph $\gio$ such that the verifier accepts $\gio$. Observe that by the feasibility property, it holds that $\gio \in\Pi$. Let $V^{L}_1,\dots V^{L}_k$ and $S^{L}_1,\dots S^{L}_k$ be the subsets $V_j$ and $S_j$ encoded by the prover in the field $L_{\text{rim}}(\cdot)$. Let $T^{L}_{j}=\inner^{2}(V^{L}_j)\cup  S^{L}_j$ and let $g^{L}_j(v)$ be the multiplicity value encoded in $L_{\text{grw}}(v)$ for each $v\in T_j$. Define $\msfo^{*}:V\rightarrow\{0,\dots ,k(\Pi,n)\}$ to be an optimal solution for $G$ and $\msfi$ and observe that $w(T^{L}_j,\msfo^{*})\leq w_{\max}(T^{L}_j)$ for every $1\leq j\leq k$. It follows from the rim property that $f(G_{\msfi,\msfo^{*}})=w(V^{L}_1,\msfo^{*})+\dots +w(V^{L}_k,\msfo^{*})=w(\bigcup_{j\in[k]}\inner ^2(V^{L}_j),\msfo^{*})+ w(\bigcup_{j\in[k]}\rim (V^{L}_j),\msfo^{*})\leq w(\bigcup_{j\in[k]}\inner ^2(V^{L}_j),\msfo^{*})+ w(\bigcup_{j\in[k]} S^{L}_j,\msfo^{*})=w(\bigcup_{j\in[k]}T^{L}_j,\msfo^{*})\leq \sum_{j\in[k]}w_{\max}(T^{L}_j)$. Now, from the optimality and growth properties, we get that $f(G_{\msfi,\msfo^{*}})\leq \sum_{j\in[k]}w_{\max}(T^{L}_j)\leq \alpha \cdot \sum_{j\in[k]}w(T^{L}_j,g^{L}_j)\leq \alpha \cdot (1+\epsilon)\cdot \sum _{j\in[k]}w(V^{L}_j,\msfo)=\alpha \cdot (1+\epsilon)\cdot f(\gio)$ which establishes the soundness requirement.
	
	\section{Compiler for CGFs}
	\label{section:compiler-cgfs}
	In this section, we present our generic compiler for CGFs. It is divided into four subsections as follows. First, in Section \ref{section:su-closed-cgfs} we characterize the CGFs that are suited for our compiler, namely SU-closed CGFs. Following that, Sections \ref{section:cgf-partition} and \ref{section:cgf-label} are dedicated to the compiler construction. More formally, these sections constructively prove the following theorem.
	
	\begin{theorem}\label{theorem:main-cgf}
		Let $\Phi$ be an SU-closed CGF that admits a PLS with a proof size of $\ell_{\Phi}$. For any constant $\delta>0$, there exists a locally restricted $\delta$-TPLS for $\Phi$ with a proof size of $\ell_{\Phi}+O(\log n)$.
	\end{theorem}

	For convenience, the compiler construction is divided between Sections \ref{section:cgf-partition}, in which we present an SLOCAL partition algorithm (that plays a similar role to the one presented in the OptDGP compiler); and \ref{section:cgf-label}, in which we describe the label assignment and verification process. In Section \ref{section:cgf-high-level}, we provide a high-level overview of the SLOCAL algorithm and how it is used in the label assignment and verification process.

	\subsection{SU-Closed CGFs}
	\label{section:su-closed-cgfs}
	A CGF $\Phi$ is said to be \emph{closed under node-induced subgraphs} if for every configured graph $G_s\in\Phi$ and node subset $U\subseteq V$, it holds that $G_s(U)\in \Phi$. We say that two configured graphs $G_s=\langle G=(V,E),s\rangle$ and $G'_{s'}=\langle G'=(V',E'),s'\rangle$ are \emph{disjoint} if $V\cap V'=\emptyset$. We define the \emph{disjoint union} between two disjoint configured graphs $G_s=\langle G=(V,E),s\rangle$ and $G'_{s'}=\langle G'=(V',E'),s'\rangle$ as the configured graph $\tilde{G}_{\tilde{s}}=\langle \tilde{G},\tilde{s}\rangle$ consisting of the graph $\tilde{G}=(V\dot{\cup}V',E\dot{\cup}E')$ and the configuration function $\tilde{s}:V\dot{\cup}V'\rightarrow \{0,1\}^{*}$ that assigns the local configuration $\tilde{s}(v)=s(v)$ to any node $v\in V$; and $\tilde{s}(v)=s'(v)$ to any node $v\in V'$. We say that a CGF $\Phi$ is \emph{closed under disjoint union} if for any two disjoint configured graphs $G_s,G'_{s'}\in\Phi$ with disjoint union $\tilde{G}_{\tilde{s}}$, it holds that $\tilde{G}_{\tilde{s}}\in\Phi$. We refer to a CGF $\Phi$ as \emph{SU-closed} if it is closed under node-induced subgraphs and under disjoint union.

	\subsection{Overview}
	\label{section:cgf-high-level}
	In Section \ref{section:cgf-partition}, we present an SLOCAL algorithm called $\partcgf$ that partitions the nodes of a given configured graph $G_s$ into clusters. Before formally describing the algorithm in Section \ref{section:cgf-partition}, let us provide some intuition by presenting the high-level idea of the partition and how it is used to design the label assignment and verification process for an SU-closed CGF $\Phi$.  
	
	Given a configured graph $G_s\in \Phi$, we use a ball growing argument to obtain a partition of the nodes into clusters such that: (1) the subgraph induced by each cluster is of logarithmic diameter; and (2) the number of crossing edges between clusters is a $\delta$-fraction of the number of edges in the clusters. 
	
	In order to allow the prover to provide a proof for the second property by local means, during $\partcgf$, some nodes may be assigned a \emph{secondary affiliation} to an adjacent cluster. The idea is that for each cluster $V_j$, the number of crossing edges to nodes with secondary affiliation to $V_j$ is a $\delta$-fraction of the number of edges within $V_j$.   
	
	\subsection{Partition Algorithm}
	\label{section:cgf-partition}
	\paragraph*{Algorithm's Description.}
	We now provide a formal description of the $\partcgf$ algorithm. Consider an SU-closed CGF $\Phi$ and a configured graph $G_s\in \Phi$. The algorithm partitions the nodes of $G$ into (possibly empty) clusters $V=V_1\dot{\cup}\dots\dot{\cup}V_n$. As usual in the SLOCAL model, the nodes are processed sequentially in $n$ iterations based on an arbitrary order $v_1,\dots,v_n$ on the nodes, where node $v_j$ is processed in the $j$-th iteration.
	
	Throughout the execution of $\partcgf$, each node $v\in V$ maintains two fields referred to as $\text{cluster}(v)$ and $\text{sec}(v)$. The field $\text{cluster}(v)$ is initially empty and its role is to identify $v$'s cluster, where each cluster $V_j$ is identified by the id of node $v_j$ which is processed in the $j$-th iteration, i.e., $V_j=\{v\mid \text{cluster}(v)=id(v_j)\}$. The field $\text{sec}(v)$ is initially empty and its role is to identify $v$'s \emph{secondary affiliation} to a cluster if such affiliation exists (otherwise it remains empty throughout the algorithm).

	The $j$-th iteration of $\partcgf$ is executed as follows. Let $G_j$ to be the subgraph induced on $G$ by $V\setminus(V_1\cup\dots\cup V_{j-1})$. If $v_j\in V_1\cup\dots\cup V_{j-1}$, then we define $V_j=\emptyset$ and finish the iteration; so, assume that $v_j$ is a node in $G_j$. For an integer $r\in \mathbb{Z}_{\geq 0}$, let $D_j^r$ be the set of nodes at distance exactly $r$ from $v_j$ in $G_j$ and let $B^r_j=\bigcup_{r'=0}^{r}D_j^{r'}$. Let $E^{r}_{j}$ be the set of edges in the subgraph $G(B^r_j)$ and let $C^{r}_{j}$ be the set of edges $(u,v)\in E$, such that $u\in B^r_j$ and  $v\notin B^r_j$. Define $r(j)$ to be the smallest integer that satisfies $|C^{r(j)}_{j}|\leq \delta \cdot |E^{r(j)}_{j}|$. The $j$-th iteration is completed by setting $\text{cluster}(v)=id(v_j)$ for each node $v\in B_j^{r(j)}$; and $\text{sec}(v)=id(v_j)$ for each node $v\in D_j^{r(j)+1}$.

	\paragraph*{Algorithm's Properties.}
	The following lemma establishes an upper bound on the diameter of each subgraph $G(V_j)$. 
	
	\begin{lemma}\label{lemma:cluter-diameter-cgf}
		The diameter of subgraph $G(V_j)$ is $O(\log n)$ for each $j\in[n]$.
	\end{lemma}
	\begin{proof}
		Suppose that $V_j\neq \emptyset$ (as the lemma is trivial otherwise). To show that the diameter of $G(V_j)$ is $O(\log n)$ it is sufficient to show that $r(j)=O(\log n)$. We use a ball growing argument. Note that for any integer $r$, it holds that $|E^{r}_{j}|\geq|E^{r-1}_{j}|+|C^{r-1}_{j}|$. Thus, for every $r'<r(j)+1$, we have
		$$|E^{r(j)+1}_{j}|\geq |E^{r'}_{j}|>(1+\delta)\cdot |E^{r'-1}_{j}|>(1+\delta)^2\cdot |E^{r'-2}_{j}|>\dots$$
		and since $n^2>m\geq |E^{r(j)+1}_{j}|$, we get that $r(j)=O((1/\delta)\log n)=O(\log n)$.
	\end{proof}
	
	A simple observation derived from Lemma \ref{lemma:cluter-diameter-cgf} is that $\partcgf$ has locality $O(\log n)$. This observation combined with the results of \cite{GhaffariKM17,RozhonG20} lead to the following corollary.
	\begin{corollary}\label{corollary:round-complexity-cgf}
		The algorithm $\partcgf$ can be simulated by a LOCAL algorithm with polylogarithmic round-complexity.
	\end{corollary}
	
	\subsection{Labels and Verification}
	\label{section:cgf-label}
	Consider an SU-closed CGF $\Phi$ and a configured graph $G_s\in\Phi$. The prover uses the SLOCAL algorithm $\partcgf$ presented in Section \ref{section:cgf-partition} to compute the values $r(j)$ for all $1\leq j\leq n$, and the fields $\text{cluster}(v)$, $\text{sec}(v)$. The goal of the prover is to provide proof of two properties satisfied by the given configured graph $G_s$ and the outcome of $\partcgf$. We refer to those properties as \emph{secondary clusters}, \emph{crossing edges}, and \emph{inclusion}. To that end, the prover produces a label assignment $L:\Vfunc$ that assigns each node $v$ with a label $L(v)=\langle L_{\text{sec}}(v), L_{\text{cross}}(v),L_{\text{inc}}(v)\rangle$. The label $L(v)$ is composed of the fields $L_{\text{sec}}(v)$, $L_{\text{cross}}(v)$, and $L_{\text{inc}}(v)$, that provide proof for the secondary clusters, crossing edges and inclusion properties, respectively.
	
	The secondary clusters property states that $\text{sec}(v)$ is not empty for every node $v$ that has a neighbor belonging to a different cluster. To that end, the sub-label $L_{\text{sec}}(v)$ assigns the values $\text{cluster}(v)$ and $\text{sec}(v)$ to each node $v\in V$. Observe that this information is sufficient for the verifier to verify the secondary clusters property.
	
	For all $j\in [n]$, let $\mathtt{sec}(V_j)=\{v\mid \text{sec}(v)=id(v_j) \}$ be the set of nodes whose secondary affiliation is to $V_j$ by the end of the $\partcgf$ algorithm, and let $F_j=\{(u,v)\in E\mid u\in V_j,v\in \mathtt{sec}(V_j)\}$ denote the set of edges with one endpoint in $V_j$ and the other endpoint in $\mathtt{sec}(V_j)$. The crossing edges property states that each cluster $V_j$ satisfies $|F_j|\leq \delta \cdot |E^{r(j)}_{j}|$. The field $L_{\text{cross}}(\cdot)$ serves the crossing edges property by means of a comparison scheme between $|F_j|$ and $\delta \cdot |E^{r(j)}_{j}|$. This comparison scheme is based on a shortest paths (spanning) tree rooted at node $v_j$ for each cluster $V_j\neq \emptyset$. Notice that each node $v\in V_j$ knows its incident edges from $|F_j|$ based on the $L_{\text{sec}}(\cdot)$ field of its neighbors.
	
	The inclusion property states that $G_s(V_j)\in\Phi$ for all $V_j\neq \emptyset$. To that end, the prover uses the field $L_{\text{inc}}(\cdot)$ to encode a PLS for $\Phi$ concurrently on all subgraphs $G(V_j)$.
	
	\paragraph*{Proof Size and Correctness.}
	We observe that the label assignment produced by the prover can be computed by means of an SLOCAL algorithm with locality $O(\log n)$ and thus it can be simulated by a locally restricted prover. Moreover, the sub-labels $L_{\text{sec}}(v)$ and  $L_{\text{cross}}(v)$ are of size $O(\log n)$; and $L_{\text{inc}}(v)$ is of size $\ell_{\Phi}$, where $\ell_{\Phi}$ is the proof size of a PLS for $\Phi$. Overall, the proof size of this scheme is $\ell_{\Phi}+O(\log n)$.
	
	We now show that the correctness requirements are satisfied. We start with completeness, i.e., showing that if $G_s\in \Phi$, then the verifier accepts $G_s$. Observe that the secondary clusters and crossing edges properties are satisfied by construction of $\partcgf$. In addition, the inclusion property follows from the fact that $\Phi$ is closed under node-induced subgraphs.
	
	As for the soundness requirement, consider a configured graph $G_s$ such that the verifier accepts $G_s$. Let $V^{L}_1,\dots V^{L}_k$ be the clusters encoded in the field $L_{\text{sec}}(\cdot)$. For every $j\in [k]$, let $E^{L}_j$ denote the edge set of subgraph $G(V^{L}_j)$, let $\mathtt{sec}^{L}_j$ be the set of nodes for which the field $L_{\text{sec}}(\cdot)$  encodes a secondary affiliation to $V^{L}_j$, and let $F^{L}_j$ be the set of edges with one endpoint in $V^{L}_j$ and one in $\mathtt{sec}^{L}_j$. The inclusion property guarantees that $G_{s}(V^{L}_j)\in\Phi$ for each $j\in[k]$. Let $G'_{s'}$ be the disjoint union of $G_{s}(V^{L}_1),\dots,G_{s}(V^{L}_k)$. Since $\Phi$ is closed under disjoint union, we get that $G'_{s'}\in\Phi$. By the secondary clusters property, $G'_{s'}$ is the configured subgraph obtained from $G_s$ by removing the set $F^{L}_1\cup\dots\cup F^{L}_k$ of edges. The crossing edges property implies that $|F^{L}_1\cup\dots\cup F^{L}_k|=\sum_{j\in[k]} |F^{L}_j|\leq \delta \cdot \sum_{j\in[k]} |E^{L}_j|\leq \delta m$. Thus, $G_s$ is not $\delta$-far from belonging to $\Phi$, i.e., $G_s\notin \calf_{N}$.
	
	In conclusion, this scheme describes a correct locally restricted $\delta$-TPLS for SU-closed CGFs with a proof size of $\ell_{\Phi}+O(\log n)$, thus proving Theorem \ref{theorem:main-cgf}. 
	
	\section{Bounds for Concrete OptDGPs and CGFs}\label{section:bounds}
	
	\subsection{OptDGPs} 
	\label{section:concrete-optdgps}
	In this section, we show how the compiler presented in Section \ref{section:compiler-optdgps} can be used in the design of locally restricted APLSs for some classical OptDGPs that fit the canonical structure. In Sections \ref{section:mwvc}, \ref{section:maxis}, and \ref{section:mwds}, we present locally restricted APLSs with a logarithmic proof size for the problems of minimum weight vertex cover, maximum independent set, and minimum weight dominating set, respectively. Then, in  Section \ref{section:generic}, we present a locally restricted $(1+\epsilon)$-APLS that applies to any canonical OptDGP.

	\subsubsection{Minimum Weight Vertex Cover} 
	\label{section:mwvc}
	Consider a graph $G=(V,E)$ associated with a node-weight function $w:V\rightarrow\{1,\dots ,n^{O(1)}\}$ and let $C\subseteq E$ be a set of \emph{constrained} edges. A \emph{vertex cover} of $C$ is a subset $U\subseteq V$ of nodes such that every edge $e\in C$ has at least one endpoint in $U$. A \emph{minimum weight vertex cover (MWVC)} of $C$ is a vertex cover $U$ of $C$ that minimizes $w(U)=\sum_{u\in U}w(u)$. 
	
	Observe that MWVC is a covering OptDGP. Moreover, vertex cover is self-induced (notice that this is in contrast to the common case of vertex cover where all edges are constrained). Thus, MWVC is canonical. We aim to use our compiler to construct a locally restricted $2(1+\epsilon)$-APLS for MWVC. To that end, we establish the following lemma.
	\begin{lemma}\label{lemma:mwvc-2-apls}
		There exists a $2$-APLS for MWVC with a proof size of $O(\log n)$.
	\end{lemma}
	\begin{proof}
		As presented in \cite{EmekG20}, there exists a $2$-APLS for the instance of MWVC where all edges are constrained and the graph is connected. We can obtain a $2$-APLS for MWVC in the more general case where a subset $C\subseteq E$ of edges are constrained simply by applying the $2$-APLS from \cite{EmekG20} on the connected subgraphs induced by the constrained edge set $C$. The proof size of this scheme is $O(\log n+\log W)$, where $W$ is an upper bound on the node-weights. Since in our case $W=n^{O(1)}$, it follows that the proof size is $O(\log n)$.
	\end{proof}
	Plugging the $2$-APLS obtained in Lemma \ref{lemma:mwvc-2-apls} into our compiler leads to the following corollary.
	\begin{corollary}\label{corollary:mwvc}
		For any constant $\epsilon>0$, there exists a locally restricted $(2(1+\epsilon))$-APLS for MWVC with a proof size of $O(\log n)$.
	\end{corollary}

	We now consider the unweighted version, simply referred to as \emph{minimum vertex cover (MVC)}, on graphs with large odd-girth (where the odd-girth of a graph is defined to be the shortest odd cycle). As the following theorem shows, this case allows for an improved approximation ratio.
	\begin{theorem}\label{theorem:mvc-odd-girth}
		For any constant $\epsilon>0$, there exists a locally restricted $(1+\epsilon)$-APLS for MVC on graphs of odd-girth $\omega(\log n)$ with a proof size of $O(\log n)$.
	\end{theorem} 
	\begin{proof}
		Recall that our compiler first partitions the nodes into clusters of diameter $O(\log n)$, and then proceeds to apply an $\alpha$-APLS concurrently on the subgraph induced by each cluster. We observe that each of these subgraphs created by the partition is bipartite since the odd-girth of the graph is $\omega(\log n)$. Thus, it is sufficient to show that there exists a PLS for MVC on bipartite graphs with a proof size of $O(\log n)$.
		
		The well known K{\"o}nig's theorem states that in bipartite graphs the size of minimum vertex cover is equal to the size of maximum matching. This allows for a PLS for MVC in bipartite graphs with a proof size of $O(\log n)$ constructed as follows. The prover simply encodes a maximum matching on the graph along with a proof that the size of this matching is equal to the size of the given vertex cover (e.g., by means of a comparison scheme). 
	\end{proof}
	\subsubsection{Maximum Independent Set} 
	\label{section:maxis}
	Consider a graph $G=(V,E)$ and let $C\subseteq E$ be a set of \emph{constrained} edges. An \emph{independent set} of $C$ is a subset $I\subseteq V$ of nodes, such that each node $v\in I$ is incident on an edge in $C$ and every edge $e\in C$ has at most one endpoint in $I$. A \emph{maximum independent set (MaxIS)} of $C$ is an independent set $I$ of $C$ that maximizes $|I|$. 
	
	Observe that MaxIS is a packing OptDGP. Moreover, independent set is self-induced (notice that this is in contrast to the common case of independent set where all edges are constrained). Thus, MaxIS is canonical. 
	
	Let us denote by $\Delta=\max_{v\in V}\{\deg(v)\}$ the largest degree in graph $G=(V,E)$. In the following lemma, we present a simple $\Delta$-APLS for the MaxIS problem. 
	\begin{lemma}\label{lemma:maxis-apls}
		There exists a $\Delta$-APLS for MaxIS with a proof size of $O(\log n)$.
	\end{lemma}
	\begin{proof}
		We use the fact that the ratio between the size of a maximum independent set and the size of a maximal independent set (i.e., an independent set that is not a subset of any other independent set) is at most $\Delta$. The $\Delta$-APLS construction is simple. The prover encodes a maximal independent set along with the value of $\Delta$ and a proof that its size is at most a multiplicative factor of $\Delta$ away from the given independent set.
	\end{proof}
	Plugging the $\Delta$-APLS from Lemma \ref{lemma:maxis-apls} into our compiler leads to the following corollary.
	\begin{corollary}
		For any constant $\epsilon>0$, there exists a locally restricted $(\Delta(1+\epsilon))$-APLS for MaxIS with a proof size of $O(\log n)$.
	\end{corollary}
	
	Similarly to the MVC problem, restricting MaxIS to families of graphs with odd-girth $\omega(\log n)$ allows for a better approximation ratio, as established by the following theorem. 
	\begin{theorem}\label{theorem:maxis-odd-girth}
		For any constant $\epsilon>0$, there exists a locally restricted $(1+\epsilon)$-APLS for MaxIS on graphs of odd-girth $\omega(\log n)$ with a proof size of $O(\log n)$.
	\end{theorem} 
	\begin{proof}
		Similarly to the proof of Theorem \ref{theorem:mvc-odd-girth}, it is sufficient to show that there exists a PLS for MaxIS on bipartite graphs with a proof size of $O(\log n)$. To that end, we can use the fact that in bipartite graphs, the size of MaxIS is equal to the size of a minimum edge cover. We can now construct a PLS where the prover encodes a minimum edge cover of the graph, along with a proof that it is equal in size to the given independent set.
	\end{proof}
	\subsubsection{Minimum Weight Dominating Set} 
	\label{section:mwds}
	Consider a graph $G=(V,E)$ associated with a node-weight function $w:V\rightarrow\{1,\dots n^{O(1)}\}$ and let $C\subseteq V$ be a subset of \emph{constrained} nodes. A \emph{dominating set} of $C$ is a subset  $D\subseteq V$ of nodes, such that $ D\cap (v\cup N(v))\neq \emptyset$ for each constrained node $v\in C$. A \emph{minimum weight dominating set (MWDS)} of $C$ is a dominating set $D$ of $C$ that minimizes $\sum_{u\in D}w(u)$. 
	
	Observe that MWDS is a covering OptDGP. Moreover, dominating set is self-induced (notice that this is in contrast to the common case of dominating set where all nodes are constrained). Thus, we get that MWDS is canonical. We aim to use our compiler to construct a locally restricted $O(\log n)$-APLS for MWDS. To that end, we establish the following lemma.
	\begin{lemma}\label{lemma:apls-mwds}
		There exists an $O(\log n)$-APLS for MWDS with a proof size of $O(\log n)$.
	\end{lemma}
	\begin{proof}
		An $O(\log n)$-APLS for the instance of MWDS where all nodes are constrained and is presented in \cite{EmekG20}. The idea behind that $O(\log n)$-APLS is that the prover provides a feasible solution to the dual LP, such that the objective value of this dual solution is at most a multiplicative factor of $O(\log n)$ from the given dominating set. We argue that this technique can be applied to obtain $O(\log n)$-APLS for MWDS (in its more generalized version described above). This follows from the fact that the  gap between an optimal MWDS solution and an optimal dual solution remains $O(\log n)$ (since MWDS is an instance of set cover).
	\end{proof}
	Plugging the $O(\log n )$-APLS obtained in Lemma \ref{lemma:apls-mwds} into our compiler leads to the following corollary.
	\begin{corollary}\label{corollary:mwds}
		There exists a locally restricted $O(\log n)$-APLS for MWDS with a proof size of $O(\log n)$.
	\end{corollary}
	\subsubsection{Generic Locally Restricted $(1+\epsilon)$-APLS for Canonical OptDGPs} 
	\label{section:generic}
	We establish a generic upper bound that applies to any canonical OptDGP.
	\begin{theorem}
		Consider a canonical OptDGP $\Psi$. For any constant $\epsilon>0$, there exists a locally restricted $(1+\epsilon)$-APLS for $\Psi$ with a proof size of $O(n^2)$.
	\end{theorem}
	\begin{proof}
		As stated in \cite[Theorem 3.2]{pls}, any decidable property admits a PLS. The idea behind this universal PLS is that the prover can assign each node $v$ with a label $L(v)$ that encodes the entire configured graph $G_s$. In response, the verifier at node $v$ verifies that $v$'s neighbors agree on the structure of the configured graph encoded in $L(v)$, and that $v$'s local neighborhood is consistent with the one encoded in $L(v)$. Following that, the verifier can evaluate whether $G_s$ is a yes-instance or not. 
		
		Observe that applying the universal PLS described above for a canonical OptDGP requires a proof size of $O(n^2)$. We can now plug this PLS construction into our compiler to obtain the desired locally restricted $(1+\epsilon)$-APLS for $\Psi$. 
	\end{proof}
	
	\subsection{CGFs}\label{section:concrete-cgfs}
	
	In this section, we show how the compiler presented in Section \ref{section:compiler-cgfs} can be combined with known PLS constructions to obtain locally restricted $\delta$-TPLSs for various well-known SU-closed CGFs.
	
	\subsubsection{Planarity}\label{section:planarity}
	A graph $G=(V,E)$ is called \emph{planar} if it can be embedded in the plane. The following lemma has been established by Feuilloley et al.~\cite{FeuilloleyFMRRT21}.
	\begin{lemma}
		There exists a PLS for planarity with a proof size of $O(\log n)$.
	\end{lemma}
	Observe that planar graphs are SU-closed. Thus, plugging the PLS for planarity into our compiler implies the following corollary.
	\begin{corollary}
		For any constant $\delta>0$, there exists a locally restricted $\delta$-TPLS for planarity with a proof size of $O(\log n)$.
	\end{corollary}

	\subsubsection{Bounded Arboricity}\label{section:arboricity}
	The \emph{arboricity} of a graph $G=(V,E)$ is the minimum number $k$ for which there exists an edge partition $E=E_1\dot{\cup}\dots \dot{\cup}E_k$ such that $G_i=(V,E_i)$ is a forest for each $i\in[k]$. Let $arb(G)$ denote the arboricity of graph $G$.  We say that graph $G$ is of bounded arboricity if $arb(G)=O(1)$.
	\begin{lemma}
		There exists a PLS for bounded arboricity with a proof size of $O(\log n)$.
	\end{lemma}
	\begin{proof}
		As established in \cite{pls}, there exists a PLS for forests with a proof size of $O(\log n)$. A PLS for bounded 
		can be implemented by using the PLS construction for forests concurrently on $arb(G)$ edge-induced subgraphs of $G$.
	\end{proof}
	Observe that bounded arboricity is SU-closed. Thus, plugging the PLS for bounded arboricity into our compiler implies the following corollary.
	\begin{corollary}
		For any constant $\delta>0$, there exists a locally restricted $\delta$-TPLS for bounded arboricity with a proof size of $O(\log n)$.
	\end{corollary}

	\subsubsection{$k$-Colorability}\label{section:k-colorability}
	For a positive integer $k$, we say that a graph $G=(V,E)$ is \emph{$k$-colorable} if there exists a proper $k$-coloring of its nodes. Observe that $k$-colorability admits a (simple) PLS with proof size $O(\log k)$ and that it is SU-closed. Thus, combined with our compiler, we get the following theorem.
	\begin{theorem}
		For any constant $\delta>0$, there exists a locally restricted $\delta$-TPLS for $k$-colorability with a proof size of $O(\log n)$.
	\end{theorem}

	\subsubsection{Forests and DAGs}\label{section:forests-dags}
	The following two lemmas have been established by Korman et al.~\cite{pls}.
	\begin{lemma}
		There exists a PLS for forests with a proof size of $O(\log n)$.
	\end{lemma}
	\begin{lemma}
		There exists a PLS for directed acyclic graphs (DAGs) with a proof size of $O(\log n)$.
	\end{lemma}
	Observe that both forests and DAGs are SU-closed. Thus, plugging the PLSs for forests and DAGs into our compiler implies the following corollaries.
	\begin{corollary}
		For any constant $\delta>0$, there exists a locally restricted $\delta$-TPLS for forests with a proof size of $O(\log n)$.
	\end{corollary}
	\begin{corollary}
		For any constant $\delta>0$, there exists a locally restricted $\delta$-TPLS for directed acyclic graphs with a proof size of $O(\log n)$.
	\end{corollary}
	
	\section{Impossibilities of Locally Restricted GPLS}
	\label{section:impossibilities}
	In this section, we establish some inherent limitations of locally restricted GPLSs based on the following observation.
	
	\begin{observation}
		If there exists a locally restricted GPLS over $\calu$ with yes-family $\calf_{Y}$ and no-family $\calf_{N}$, then there exists a LOCAL algorithm with a $\log ^{O(1)}(n)$ round-complexity that given a configured graph $G_s\in\calu$, decides if $G_s\in \calf_{Y}$ (in which case all nodes return $\true$); or $G_s\in \calf_{N}$ (in which case at least one node returns $\false$).
	\end{observation}
	\begin{proof}
		Given a configured graph $G_s\in\calu$, we obtain a LOCAL algorithm by first simulating the locally restricted prover on $G_s$ (using a polylogarithmic number of rounds), and then simulating the verifier (using $1$ round). By the correctness requirements of a GPLS, the outcome of this algorithm is that all nodes return $\true$ if $G_s\in \calf_{Y}$; whereas at least one node returns $\false$ if $G_s\in \calf_{N}$.
	\end{proof}
	
	The observation above implies that it is impossible to construct a locally restricted GPLS for verification tasks that require $\omega( poly\log n)$ rounds in the LOCAL model. Notice that this impossibility applies to a large class of verification tasks associated with OptDGPs and CGFs. For example, using a simple indistinguishability argument, one can show that there is no locally restricted PLS for forests (i.e., a PLS deciding if a given graph is a forest). Similar arguments can be applied to exclude a locally restricted PLS for most of the OptDGPs and CGFs considered in Section \ref{section:bounds}.


	\section{Additional Related Work}
	\label{section:related}
	The PLS model was introduced by Korman, Kutten, and Peleg in \cite{pls} and studied extensively since then, see, e.g., \cite{FeuilloleyFMRRT21,mst,blin-2014, error-sensitive-pls-Fraigniaud-2017, patt-shamir-perry-2017,redundancy-Feuilloley-2018}. Research in this field include \cite{mst}, where a PLS for minimum spanning tree  is shown to have a proof size of $O(\log n\log W)$, where $W$ is the largest edge-weight, and \cite{FeuilloleyFMRRT21}, where a PLS for planarity is shown to have a proof size of $O(\log n)$. 
	
	In parallel, several researchers explored the limitations of the PLS model, often relying on known lower bounds from nondeterministic communication complexity \cite{kushilevitz-communication-complexity}. Lower bounds of $\Omega(n^2)$ and $\Omega(n^2/\log n)$ are established in \cite{lcp} with regards to the proof size of any PLS for graph symmetry and non $3$-colorability, respectively. A similar technique was used by the authors of \cite{lower_bounds} to show that many classic optimization problems require a proof size of $\tilde{\Omega}(n^2)$.
	
	The lower bounds on the proof size of PLSs for some optimization problems
	have motivated the authors of \cite{aplspaper} to introduce the APLS notion, further studied recently in \cite{EmekG20}. Optimization problems considered in the context of APLS include maximum weight matching, which was shown in \cite{aplspaper} to admit a $2$-APLS with a proof size of $O(\log W)$, and minimum weight vertex cover, which was shown in \cite{EmekG20} to admit a $2$-APLS with a proof size of $O(\log n+\log W)$, where in both cases $W$ refers to the largest weight value. 
	
	In the current paper, we also introduce the TPLS model which is suited for properties that are not formulated as optimization problems. This model is based on the notion of property testing \cite{GoldreichGR98}. More specifically, the TPLS model is formulated using the distance measure between graphs defined in \cite{AlonKKR08}. 
	
	Our focus in this paper is on locally restricted APLSs and TPLSs, restricting the prover to a LOCAL algorithm with a polylogarithmic number of rounds. Interest in the power of deterministic LOCAL algorithms with polylogarithmic round-complexity was initiated by Linial's seminal work \cite{Linial87,Linial92}. One particular problem that raised a lot of interest in this context is the network decomposition problem introduced by Awerbuch et al.\ in \cite{AwerbuchGLP89}. In a recent breakthrough \cite{RozhonG20}, Ghaffari and Rozhon presented a deterministic algorithm with polylogarithmic round-complexity for the network decomposition problem. As established in \cite{GhaffariKM17}, a consequence of this result
	is that any SLOCAL algorithm with
	$\log^{O (1)} n$
	locality can be simulated by a LOCAL algorithm with $\log ^{O(1)}n$ rounds. This simulation technique is used in the construction of our compilers in Sections \ref{section:compiler-optdgps} and \ref{section:compiler-cgfs}.

	\newpage
	\begin{table}
		\begin{center}
			\begin{tabular}{| l| l| l|}
				\hline
				\textbf{Term} &  \textbf{Abbreviation} & \textbf{Reference}\\
				\hline
				configured graph family & CGF& Section \ref{section:model}\\
				\hline
				closed under node-induced subgraph and disjoint union & SU-closed & Section \ref{section:model}\\
				\hline
				distributed graph problem & DGP& Section \ref{section:model}\\
				\hline
				input-output graph & IO graph& Section \ref{section:model}\\
				\hline
				distributed graph minimization problem & MinDGP& Section \ref{section:model}\\
				\hline
				distributed graph maximization problem & MaxDGP& Section \ref{section:model}\\
				\hline
				distributed graph optimization problem & OptDGP& Section \ref{section:model}\\
				\hline
				gap proof labeling scheme & GPLS& Section \ref{section:pls}\\
				\hline
				proof labeling scheme & PLS& Section \ref{section:pls}\\
				\hline
				testing proof labeling scheme & TPLS& Section \ref{section:pls}\\
				\hline
				approximate proof labeling scheme & APLS& Section \ref{section:pls}\\
				\hline
				
			\end{tabular}
		\end{center}
		\caption{A list of abbreviations}
		\label{abbreviations}
	\end{table}

	\clearpage
	\bibliographystyle{alpha}
	
	\bibliography{references}
	
\end{document}